\documentclass[11pt]{amsart}
\usepackage{graphicx}
\usepackage{epstopdf,epsfig}
\usepackage[utf8]{inputenc}
\usepackage[english]{babel}
\usepackage[T1]{fontenc}
\usepackage{amsthm,amsmath,amssymb,amsfonts,amscd,epsfig}
\usepackage{xcolor}
\usepackage{enumitem, mathtools}
\usepackage{pgfplots}
\pgfplotsset{compat=1.16}
\usetikzlibrary{arrows.meta,matrix}

\theoremstyle{plain}
\newtheorem{theorem}{Theorem}[section]

\newtheorem{proposition}[theorem]{Proposition}
\newtheorem{lemma}[theorem]{Lemma}

\theoremstyle{definition}
\newtheorem{definition}[theorem]{Definition}

\theoremstyle{remark}
\newtheorem{remark}[theorem]{Remark}

\begin{document}

\begin{abstract}
In this article we analyze several mathematical models with singularities where the classical cotangent model is replaced by a $b$-cotangent model. We provide physical interpretations of the singular symplectic geometry underlying in $b$-cotangent bundles featuring two models: the canonical (or non-twisted) model and the twisted one. The canonical one models systems on manifolds with boundary and the twisted one represents Hamiltonian systems with a singularity on the fiber. The twisted cotangent model includes (for linear potentials) the case of fluids with dissipation. We prove (non)-existence of cotangent lift dynamics and show the existence of an infinite number of escape orbits in this model. We also discuss more general physical interpretations of the twisted and non-twisted $b$-symplectic models. Twisted $b$-symplectic models yield in a natural way escape orbits that go to the critical set. Under compactness assumptions those escape orbits are continued as \emph{singular periodic orbits} in the sense of \cite{MirandaOms21} and \cite{Miranda20}. These models offer a Hamiltonian formulation for systems which are dissipative, extending the horizons of Hamiltonian dynamics and opening a new approach to study non-conservative systems.
\end{abstract}

\author{Baptiste Coquinot}
\address{Baptiste Coquinot, Laboratoire de Physique de l'\'Ecole Normale Sup\'erieure, ENS, Universit\'e PSL, CNRS, Sorbonne Universit\'e, Universit\'e Paris Cit\'e, 24 rue Lhomond, 75005 Paris, France}
\email {baptiste.coquinot@ens.fr}
\author{Pau Mir}
\address{Pau Mir,
Laboratory of Geometry and Dynamical Systems, Universitat Polit\`{e}cnica de Catalunya, Avinguda del Doctor Mara\~{n}on 44-50, 08028, Barcelona }
\email {pau.mir.garcia@upc.edu}
\author{Eva Miranda}
\address{Eva Miranda,
Laboratory of Geometry and Dynamical Systems $\&$ Institut de Matem\`atiques de la UPC-BarcelonaTech (IMTech), Universitat Polit\`{e}cnica de Catalunya, Avinguda del Doctor Mara\~{n}on 44-50, 08028, Barcelona \\ CRM Centre de Recerca Matem\`{a}tica, Campus de Bellaterra
Edifici C, 08193 Bellaterra, Barcelona }
\email{eva.miranda@upc.edu}

\thanks{B. Coquinot is funded by the J.-P. Aguilar grant of the CFM Foundation}
\thanks{P. Mir is funded in part by the Doctoral INPhINIT - RETAINING grant ID 100010434 LCF/BQ/DR21/11880025 of “la Caixa” Foundation}
\thanks{P. Mir and E. Miranda are partially supported by the AEI grant PID2019-103849GB-I00 of MCIN/ AEI /10.13039/501100011033}
\thanks{E. Miranda is supported by the Catalan Institution for Research and Advanced Studies via an ICREA Academia Prizes 2016 and 2021 and by the Spanish State Research Agency, through the Severo Ochoa and Mar\'{\i}a de Maeztu Program for Centers and Units of Excellence in R\&D (project CEX2020-001084-M). E. Miranda also acknowledges partial support from the grant “Computational, dynamical and geometrical complexity in fluid dynamics”, Ayudas Fundación BBVA a Proyectos de Investigación Científica 2021.}

\title{Singular cotangent models in fluids with dissipation}

\maketitle

\section{Introduction}

The study of fluid mechanics has a long and rich history, revealing a complex structure on both the physical and the mathematical levels. We point to recent work detailing how new geometric facets of this complexity have been revealed through several reincarnations (see \cite{CardonaMirandaPeralta-SalasPresas21,CardonaMirandaPeralta-Salas21}). As it is well-known, in the Navier-Stokes equation the Reynolds number provides a measure of fluid complexity, giving rise to turbulence for high Reynolds number flows (with infinite Reynolds number corresponding to the Euler flow). The present work is not specifically about fluid mechanics, but an aligned investigation of the singular geometric nature of the case of a $0$ (or very low) Reynolds number flow, corresponding to a laminar flow, expressed in terms of a finite-dimensional analogy.

In this article we give a mathematical interpretation of the physics of fluids obeying the Stokes' Law using the Hamiltonian formalism on a singular cotangent model. However, the inherent geometry of this system does not let us identify this model as a twisted cotangent lift in general.

Symplectic geometry provides the landscape where classical mechanics take place. The pair of position and momenta is the physical manifestation of the existence of a cotangent bundle underlying this picture. The role of the base and fibers of the cotangent bundle is an important landmark that fixes and makes precise Hamiltonian dynamics. However, this perfect symplectic picture is often insufficient to describe the complexity of physical phenomena. Poisson geometry provides a more general scenery appropriate to capture the intricacy of physical systems. Nevertheless, Poisson geometry is, in general, too involved and even the existence of appropriate local coordinates is a difficult battleground. From this perspective, singular symplectic manifolds provide a much more controlled terrain to fulfill some of these needs. In this article we explore some physical systems that can be described as singular symplectic manifolds. We focus on the class of $b$-symplectic manifolds and identify two models: a canonical and a twisted one. We associate relevant physical systems to these two models.

In the context of symplectic geometry, singular forms have been an important object of study in the last years. A main class of such singular forms is the class of $b$-symplectic forms, formally introduced in \cite{GuilleminMirandaPires11} and \cite{GuilleminMirandaPires14}. They provide a way to model systems with boundary and to study manifolds through compactification.

Among the variety of applications of $b$-symplectic geometry (and its $b$-contact counterpart), there have been obtained remarkable results on general integrable systems, celestial mechanics and fluid dynamics (see, for instance, \cite{MirandaOms18}, \cite{MirandaOms21}, \cite{PollardAlexander19}, \cite{BraddellDelshamsMirandaOmsPlanas19}, \cite{CardonaMirandaPeralta-Salas19} \cite{DelshamsKiesenhoferMiranda17}).

The phase space of a physical problem can be associated with the cotangent bundle of the configuration space. Therefore, it is automatically symplectic and this is one of the main reasons that makes symplectic geometry the natural language of mechanics. In the general setting, the physical Hamiltonian is the sum of a kinetic term depending on the momentum and a potential term depending only on the position. It provides an associated Hamiltonian flow which describes exactly the physical dynamics and yields the usual Newton's laws.

At the crossroads of $b$-symplectic techniques and cotangent models for physical systems, \textit{singular cotangent models} supply the techniques to generalize procedures such as the cotangent lift from symplectic manifolds to $b$-symplectic manifolds. These techniques were explored in \cite{KiesenhoferMiranda17} in the integrable case and following two different approaches. In the first approach, the singularity of the integrable system defined in a cotangent bundle of a smooth manifold is placed at the base manifold. In the second case, the so-called \textit{twisted case}, the singularity is placed at the fiber. In both cases, the singularity permeates the geometric structure and  the $b$-symplectic form  carries the characterization of the singularity. Singular cotangent models have also been considered in \cite{CardonaMirandaPeralta-Salas19}, \cite{MirandaOms18} and \cite{MirandaOms21}. Connections to other singularities in physical systems are explored in \cite{YoshidaMorrison20} (see also \cite{NestTsygan01} and \cite{MirandaScott21} for the geometrical study of a more general class of singular structures called $E$-symplectic structures).

In this article we give a new application of the twisted cotangent model. In particular, we present the Stokes' Law of motion for free-falling particles in fluids with viscosity as a twisted cotangent model. We prove that, in general, a one-dimensional motion with a dissipation which is proportional to the velocity can be modeled by a twisted $b$-symplectic form.

The fact that $b$-symplectic techniques can be used to model fluid systems is interesting because, classically, the study of the evolution of moving fluids has been tackled via partial differential equations such as the Navier-Stokes equations. A set of PDEs can model viscous Newtonian fluids expressing their mass and momentum conservation but, usually, solutions to these PDEs can only be found numerically. Besides, in general it has been difficult to prove if, for some initial conditions, they are smooth or even continuous. This complexity has led to other approaches to model the behaviour of fluids. Among the alternative formulations of fluid dynamics, there are the Hamiltonian and Lagrangian formulations, which are used naturally in a wide collection of mechanics problems. In this respect, the $b$-symplectic approach given in this paper contributes to this alternative approach.

In \cite{Morrison86}, Morrison introduced the metriplectic formalism as an extension of the Hamiltonian formalism so as to include dissipation while maintaining a conserved energy-like quantity. This formalism couples Poisson brackets, coming from the Hamiltonian symplectic formalism, with metric brackets, coming from out-of-equilibrium thermodynamics (see also \cite{Morrison84a}, \cite{Morrison84b}, \cite{Morrison98}, \cite{MaterassiMorrison17} and \cite{CoquinotMorrison20}). Thus, the formalism describes systems with both Hamiltonian and dissipative components that can model friction, electric resistivity, collisions and more, in various contexts ranging over biophysics, geophysics, and plasma physics. The construction builds in asymptotic convergence to a pre-selected equilibrium state.

Following these ideas, in this article we make use of Hamilton's equations to model a system which is dissipative in the classical sense. The original idea is that we do not rearrange the conservative Hamilton's equations but, instead, we introduce a singularity at the level of the symplectic structure of the manifold, which we equip with a twisted $b$-symplectic form.

\subsection*{Organization of this article} In Section \ref{section:preliminaries} we give a crash course on $b$-symplectic geometry. In Section \ref{section:bsymplecticmodel} we introduce the new model for fluids with dissipation based on a twisted $b$-symplectic structure. We start with the one-dimensional case and the linear potential, which provides an analogue of the Stokes' Law, and we extend it to higher dimensions and more general potentials. We observe the existence of escape orbits in the twisted model and prove that the dynamics of our model does not come from the cotangent lift of a group action. In Section \ref{section:timedependentmodels} we consider time-dependent singular models in which friction arises from a re-scaling of time. Finally, in Section \ref{section:conclusions} we summarize the results of the paper and present our conclusions.

\subsection*{Acknowledgements}

This collaboration started with an internship of Baptiste Coquinot at the Laboratory of Geometry and Dynamical Systems. The visit caught us at the beginning of the pandemics so it was finally virtual. Baptiste Coquinot would like to acknowledge the (numerical) hospitality of the Laboratory of Geometry and Dynamical Systems of the Universitat Politècnica de Catalunya and especially Eva Miranda for her supervision.

\section{Preliminaries}
\label{section:preliminaries}

\subsection{$b$-Symplectic geometry}

A symplectic manifold is a manifold $M$ which admits a \emph{symplectic form $\omega$} which is closed and non-degenerate $2$-form. Given a function $H$ over a symplectic manifold, called Hamiltonian, it is useful to consider its associated \emph{Hamiltonian flow}, which is the flow of the vector field $X$ defined by $\iota_X \omega=-dH$. The existence and uniqueness of $X$ and its flow are a consequence of the non-degeneracy of the symplectic form.

In physics, the usual and more general formalism used to study dynamics is Poisson geometry \cite{MarsdenRatiu99}). Poisson manifolds are generalizations of symplectic manifolds in which the symplectic form $\omega$ is replaced by a bivector $\Pi$. Indeed, a symplectic form $\omega$ in a symplectic manifold $(M,\omega)$ may be seen as a smooth map from the space of vector fields $\mathfrak{X}(M)$ to the space of 1-forms $\Omega^1(M)$. Among the large class of Poisson manifolds we find $b$-symplectic manifolds, that can also be considered a wider class of manifolds which contains symplectic manifolds.

The basic definitions of $b$-symplectic geometry start with the notions of \textit{$b$-manifold} (a pair $(M,Z)$ where $Z$ is a hypersurface in a manifold $M$), \textit{$b$-map} (a map $f:(M_1,Z_1) \longrightarrow (M_2, Z_2)$ between $b$-manifolds with $f$ transverse to $Z_2$ and $Z_1 = f^{-1}(Z_2)$) and \textit{$b$-vector field} (a vector field on $M$ which is tangent to $Z$ at all points of $Z$).

Let $(M^n,Z)$ be a $b$-manifold. If $x$ is a local defining function for $Z$ on an open set $U\subset M$ and $(x,y_1,\ldots,y_{n-1})$ is a chart on $U$, then the set of $b$-vector fields on $U$ is a free $C^\infty(M)$-module with basis
$$(x {\frac{\partial}{\partial x}}, {\frac{\partial}{\partial y_1}},\ldots, {\frac{\partial}{\partial y_n-1}}).$$

There exists a vector bundle associated to this module called \textit{$b$-tangent bundle} and denoted by $^b TM$. The \textit{$b$-cotangent bundle} $^b T^*M$ of $M$ is defined to be the vector bundle dual to $^b TM$.

For each $k>0$, let $^b\Omega^k(M)$ denote the space of sections of the vector bundle $\Lambda^k(^b T^*M)$, which are called \textit{$b$-de Rham $k$-forms}. For any defining function $f$ of $Z$, every $b$-de Rham $k$-form can be written as
\begin{equation}
\omega=\alpha\wedge\frac{df}{f}+\beta, \text{ with } \alpha\in\Omega^{k-1}(M) \text{ and } \beta\in\Omega^k(M).
\end{equation}

A special class of closed $b$-de Rham $2$-forms is the class of \textit{$b$-symplectic forms} as defined in \cite{GuilleminMirandaPires14}. It contains forms with singularities and can be introduced formally for $b$-symplectic manifolds, making it possible to extend the symplectic structure from $M\backslash Z$ to the whole manifold $M$.

\begin{definition}[$b$-symplectic manifold]
Let $(M^{2n},Z)$ be a $b$-manifold and $\omega\in\,^b\Omega^2(M)$ a closed $b$-form. We say that $\omega$ is \textit{$b$-symplectic} if $\omega_p$ is of maximal rank as an element of $\Lambda^2(\,^b T_p^* M)$ for all $p\in M$. The triple $(M,Z,\omega)$ is called a \textit{$b$-symplectic manifold}.
\label{def:bsymplecticmanifold}
\end{definition}

\subsection{$b$-cotangent lifts}
\label{subsec:bcotangentlifts}

The cotangent bundle of a smooth manifold $M$ is naturally equipped with a symplectic structure, since there is always an intrinsic canonical linear form $\lambda$ on $T^{*}M$ defined by
\begin{equation*}
    \langle \lambda_p,v\rangle=\langle p,d\pi_pv\rangle, \hspace{25pt} p=(m,\xi)\in T^{*}M,v\in T_p(T^{*}M),
\end{equation*}
where $d\pi_p:T_p(T^{*}M)\longrightarrow T_mM$ is the differential of the canonical projection at $p$. In local coordinates $(q_i,p_i)$, the form is written as $\lambda=\sum_i p_i\,dq_i$ and is called the \textit{Liouville $1$-form}. Its differential $\omega=d\lambda=\sum_i dp _i\wedge dq_i$ is a symplectic form on $T^{*}M$.

For $b$-symplectic manifolds there are two natural choices for the singular Liouville form, each of them giving a different symplectic form $\omega$:

\begin{enumerate}
    \item Non-twisted forms: $\lambda=\frac{c}{q_1}p_1dq_1+\sum_{i=2}^n p_idq_i$ and $\omega=\frac{c}{q_1}dp_1\wedge dq_1 + \sum_{i=2}^n dp_i\wedge dq_i$,
    \item Twisted forms: $\lambda=c\log(p_1)dq_1+\sum_{i=2}^n p_idq_i$ and $\omega=\frac{c}{p_1}dp_1\wedge dq_1 + \sum_{i=2}^n dp_i\wedge dq_i$.
    \label{def:bsymplecticforms}
\end{enumerate}

The non-twisted, or canonical, symplectic form carries the singularity at the base (the transversal hypersurface $Z$ is given by $q_1=0$), while the twisted symplectic form carries the singularity at the fiber (the transversal hypersurface $Z$ is given by $p_1=0$). The constant $c$ in the expression of the forms is called the \textit{modular weight}.

The cotangent lift of a group action is defined in the following way.

\begin{definition}
Let $\rho:G\times M \longrightarrow M$ be a group action of a Lie group $G$ on a smooth manifold $M$. For each $g\in G$, there is an induced diffeomorphism $\rho_g:M\longrightarrow M$. The \textit{cotangent lift of $\rho_g$}, denoted by $\hat{\rho}_g$, is the diffeomorphism on $T^{*}M$ given by
\begin{equation*}
    \hat{\rho}_g(q,p):=(\rho_g(q),((d{\rho_g)}_q^{*})^{-1}(p)),\hspace{25pt}\text{with }(q,p)\in T^*M,
\end{equation*}
which makes the following diagram commute (see also Figure \ref{fig:cotangentlift}):
\begin{center}
\begin{tikzpicture}
  \matrix (m) [matrix of math nodes,row sep=4em,column sep=4em,minimum width=2em]
  {T^*M & T^*M \\
   M & M\\};
  \path[-stealth]
    (m-1-1) edge node [right] {$\pi$} (m-2-1)
            edge node [above] {$\hat{\rho}_g$} (m-1-2)
    (m-2-1) edge node [above] {${\rho}_g$} (m-2-2)
    (m-1-2) edge node [right] {$\pi$} (m-2-2);
\end{tikzpicture}
\end{center}
\end{definition}

Given a diffeomorphism $\rho:M\longrightarrow M$, its cotangent lift is defined in an analogous way and it preserves the Liouville $1$-form $\lambda$. As a consequence, it also preserves the symplectic form on $T^*M$. In the twisted case, the twisted $b$-cotangent lift preserves the twisted $1$-form $\lambda=c\log(p_1)dq_1+\sum_{i=2}^n p_idq_i$ and the twisted $b$-symplectic form $\omega=\frac{c}{p_1}dp_1\wedge dq_1 + \sum_{i=2}^n dp_i\wedge dq_i$.

\begin{figure}[ht!] 
\begin{tikzpicture}[scale=1.0]
\newcommand{\aplane}[2]{
	(#1, #2, 1.35) --
	++(3, 1, 0.0) --
	++(1, -1.8, -2.7) --
	++(-3, -1, -0.0) --
	cycle}
\newcommand{\bplane}[2]{
	(#1, #2, 0) --
	++(3, -1, 0.0) --
	++(1, -1, -3.5) --
	++(-3, 1, -0.0) --
	cycle}
\coordinate (adir) at (0.5, 0.5, 0.0);  
\coordinate (bdir) at (0.8, -0.15, 0);
\newcommand{\evec}[2]{#1 -- ++#2}
\newcommand{\qvec}[2]{#1 -- ++#2}
\coordinate (a) at (0.5,2,0);
\coordinate (ad) at (0.5,4,0);
\coordinate (b) at (6.5,2,0);
\coordinate (bd) at (6.5,4,0);
\coordinate (ab) at (3.5,1.5,0);
\coordinate (bvec) at (6.9, 1.93, 0);
\coordinate (avec) at (0.75,2.25,0);
\coordinate (bdvec) at (6.9, 3.93, 0);
\coordinate (advec) at (0.75,4.25,0);
\coordinate (abs) at (3.5,3,0);
\coordinate (abds) at (3.5,5,0);
\coordinate (A) at (0,0,0);
\coordinate (B) at (10,0,0);
\coordinate (C) at (9,3,1);
\coordinate (D) at (-3,2,1);
\node at (5,0.5,0) {$ $}; 
\node at (-0.4,1,0) {$ $}; 
\node at (9,1.2,0) {$ $}; 
\node at (-0.4,3,0) {$ $}; 
\node at (9,3.2,0) {$ $}; 
\draw[thick, color=black] (A) to [bend left=20] (B);
\draw[thick, color=black] (A) to [bend left=-20] (D);
\draw[thick, color=black] (D) to [bend left=20] (C);
\draw[thick, color=black] (C) to [bend left=20] (B);
\draw[fill=gray!10]\aplane{-1.3}{2.4}; 
\draw[fill=gray!10]\bplane{4}{2.3};
\draw[fill=blue!10]\aplane{-1.3}{4.4}; 
\draw[fill=blue!10]\bplane{4}{4.3};
\draw[->, thick]\evec{(a)}{(adir)} node[anchor=east] {$(q,p)$};
\draw[->, thick]\evec{(ad)}{(adir)} node[anchor=east] {$ $};
\draw[dotted] (a) -- (ad);
\draw[->, thick]\qvec{(b)}{(bdir)}; node[anchor=east] {$ $};
\draw[->, thick]\qvec{(bd)}{(bdir)} node[anchor=west] {$ $};
\draw[dotted] (b) -- (bd);
\draw[color=blue, <-] (b) to [bend left=10] (ab);
\draw[color=blue, <-] (ab) node[anchor=north] {$\rho_g$} to [bend left=10] (a);
\draw[color=red, <-] (bvec) to [bend left=-10] (abs);
\draw[color=red, <-] (abs) node[anchor=north] {$(d\rho_g)_q$} to [bend left=-10] (avec);
\draw[color=red, ->] (bdvec) to [bend left=-10] (abds);
\draw[color=red, ->] (abds) node[anchor=north] {$(d\rho_g)_q^*$} to [bend left=-10] (advec);	
\draw[fill] (b) circle[radius=2pt] node[anchor=east] {$ $};
\draw[fill] (bd) circle[radius=2pt] node[anchor=east] {};
\draw[fill] (a) circle[radius=2pt] node[anchor=east] {$ $};
\draw[fill] (ad) circle[radius=2pt] node[anchor=east] {};
\end{tikzpicture}
\caption{The cotangent lift of an action $\rho_g$ is a map on the cotangent bundle $T^*M$.}
\label{fig:cotangentlift}
\end{figure}
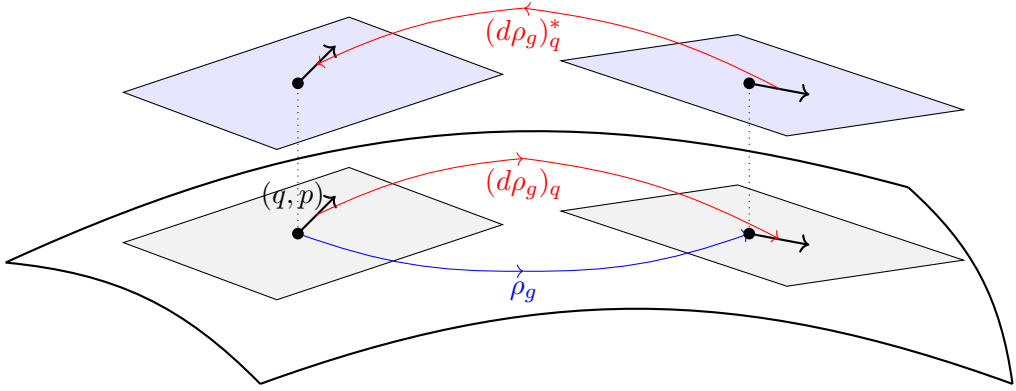

\begin{proposition}[Kiesenhofer-Miranda, \cite{KiesenhoferMiranda17}]
Given a group action $\rho:G\times M\to M$ on a smooth manifold of dimension $n$, the twisted $b$-cotangent lift $\hat\rho$ is $b$-Hamiltonian with moment map $\mu:T^*M \to \mathfrak{g}^*$ given by $$\langle\mu(p),X \rangle := \langle \lambda_p ,X^\#|_{p} \rangle =\langle p,X^\#|_{\pi(p)}\rangle,$$
where $p\in T^*M$, $X$ is an element of the Lie algebra $\mathfrak{g}$ and $X^\#$ denotes the fundamental vector field of $X$ generated by the action on $T^*M$.

Moreover, for a toric action, the moment map of the lifted action with respect to the twisted $b$-symplectic form $\omega=\frac{c}{p_1} d\theta_1\wedge dp_1+ \sum_{i=2}^n d\theta_i\wedge dp_i$ on $T^*M$ is given by $\mu=(c \log|p_1|,p_2, \dots,p_n)$.
\label{prop:cotangentlifthamiltonian}
\end{proposition}

\subsection{Cotangent models for integrable systems}
\label{subsec:cotangentmodels}
The dynamics of an integrable system $F=(f_1,\dots,f_n)$ defined in a manifold $M$ is explained by the classical Arnold-Liouville-Mineur Theorem at the regular points of $F$, namely, at the points of $M$ where the differential $dF=(df_1, \dots,df_n)$ is not singular. This theorem was restated by Kiesenhofer and Miranda in \cite{KiesenhoferMiranda17} revealing that at a semilocal level the regular leaves are equivalent to a completely toric cotangent lift model.

\begin{theorem}[Kiesenhofer-Miranda, \cite{KiesenhoferMiranda17}]
Let $F=(f_1,\dots,f_n)$ be an integrable system on a symplectic manifold $(M,\omega)$. Then, semilocally around a regular Liouville torus, the system is equivalent to the cotangent model $(T^* \mathbb{T}^n)_{can}$ restricted to a neighbourhood of the zero section $(T^* \mathbb{T}^n)_0$ of $T^* \mathbb{T}^n$.
\label{thm:cotangentliftALM}
\end{theorem}

Cotangent lifts arise naturally in physical problems and the link between integrable systems and cotangent models is clear in view of Theorem \ref{thm:cotangentliftALM}. For the singular cases, and, in particular, for $b$-integrable systems, cotangent models can be made explicit in terms of action-angle coordinates. 

\begin{theorem}[Kiesenhofer-Miranda-Scott \cite{KiesenhoferMirandaScott16}]
Suppose $(M, Z, \omega, F)$ is a $b$-integrable system and let $m \in Z$ be a regular point of $F$ for which the integral manifold containing $m$ is compact, i.e. a Liouville torus. Then, there exists an open neighbourhood $U$ of the torus and coordinates $(\theta_1,\dots,\theta_n,p_1,\dots,p_{n}): U \to \mathbb T^n \times B^n$ such that
\begin{equation}
        \omega|_U =\sum_{i=1}^{n-1} dp_i \wedge d\theta_i  + \frac{c}{p_n} dp_n \wedge d\theta_n,
\end{equation}
where the coordinates $p_1,\dots,p_n$ depend only on $F$ and the constant $c$ is the modular weight of the component of $Z$ containing $m$.
\label{thm:b-actionangle}
\end{theorem}

\section{The twisted $b$-symplectic model for dissipation}
\label{section:bsymplecticmodel}

In this section, we describe how $b$-symplectic geometry offers a way to model, in a Hamiltonian fashion, a particle moving in a dissipative fluid with viscosity. In particular, we construct an example that uses the twisted $b$-symplectic form in the cotangent bundle of $\mathbb{R}$. This example gives precisely the equation of the friction drag force exerted on a small spherical particle moving through a viscous laminar fluid in one dimension, the so-called Stokes' Law. Then, we generalize this model to higher dimensions and to other configuration spaces different from $\mathbb{R}^n$.

Take $M=\mathbb{R}$ and $T^*M\cong \mathbb{R}^2$ with coordinates $(q,p)$. Consider the Hamiltonian
\begin{equation}
    H(q,p)=\frac{p^2}{2}+f(q),
\end{equation}
which corresponds to the energy of a massive particle subject to a potential $f(q)$. The Hamilton's equations derived from $\iota_{X_H}\omega=-dH$ with the standard symplectic form $\omega=dp\wedge dq$ provide the following system, which models the main toy models in classical mechanics:
\begin{equation}
    \begin{cases}
    \Dot{q}=p\\
    \Dot{p}=-\frac{\partial f}{\partial q}
    \end{cases}.
\label{eq:Hameqs2}
\end{equation}

But, more interestingly, the Hamilton's equations derived from $\iota_{X_H}\,\omega=-dH$ with the twisted $b$-symplectic form
$$\omega=\frac{1}{p}dp\wedge dq$$
are:
\begin{equation}
    \begin{cases}
    \Dot{q}=p^2\\
    \Dot{p}=-p\frac{\partial f}{\partial q}
    \end{cases}.
\label{eq:Hameqs}
\end{equation}

At $p=0$ there are just fixed points and system (\ref{eq:Hameqs}) gives no dynamics. Hence, we can reduce the dynamical study to $p>0$, and for $p<0$ it will be symmetric up to a change of sign.

Differentiating the first equation of system (\ref{eq:Hameqs}) and substituting into the second one, we find
\begin{equation}
    \ddot{q}=-2\dot{q}\frac{\partial f}{\partial q},
    \label{eq:secondorderODE}
\end{equation}
which is a second order ODE depending only on $q$. Notice that, although we have associated $q$ to the position coordinate, $\dot q$ is not equal to the standard physical momentum $p$ but to $p^2$. However, we can still think of $p=\sqrt{\dot{q}}$ as a modified physical momentum, since it is an increasing function of $\dot q$. Taking into account this point of view, we proceed to obtain various models of dynamics for different families of potentials $f(q)$.

\subsection{A new model of Hamiltonian dissipation. The Stokes' Law as a twisted model}

A natural choice for the potential $f(q)$ is a function of linear type. This simple model already gives an original way of considering dissipation as a $b$-symplectic model, as the following result proves.

\begin{theorem}[Dissipation as a twisted singular cotangent model]
Consider the twisted $b$-symplectic model in $T^*\mathbb{R}$, given by Equation (\ref{eq:Hameqs}). The particular case $f(q)=\frac{\lambda}{2}q$ corresponds to the model of a spherical particle moving in a fluid with viscosity and suffering a friction proportional to its velocity, i.e., to the Stokes' Law.
\label{th:dissipationisacotangentmodel}
\end{theorem}

\begin{proof}
Consider $f(q)=\frac{\lambda}{2}q$, with $\lambda>0$, in the case of the Hamilton's equations coming from the twisted $b$-symplectic form, namely, in system (\ref{eq:Hameqs}). Explicitly, Hamilton's equations are
\begin{equation}
    \begin{cases}
    \Dot{q}=p^2\\
    \Dot{p}=-\frac{\lambda}{2}p
    \end{cases}.
\label{eq:Hameqsf}
\end{equation}

The corresponding second order ODE becomes

\begin{equation}
    \ddot{q}=-\lambda\dot{q},
    \label{eq:secondorderODE2}
\end{equation}

which corresponds exactly to the equation of a free massive particle moving in one dimension and affected by viscous friction. In fact, the Stokes' Law (\ref{eq:StokesLaw}) describes precisely the same case, which appears in the study of non-ideal fluids. It states that the frictional force $F$ is:

\begin{equation}
    F=6\pi \mu R v,
    \label{eq:StokesLaw}
\end{equation}

\noindent where $\mu$ is the dynamic viscosity, $R$ is the radius of the particle and $v$ is the flow velocity relative to the object (or minus the object velocity relative to the flow). The Stokes' Law computes the magnitude of the drag force that is acting against the particle motion and slowing it. This force is proportional to the velocity of the particle with respect to the fluid and of opposite direction.

Denoting the velocity $v$ by $\dot{q}$, assuming that the force $F$ is proportional to the acceleration $\ddot{q}$ and combining physical constants, we deduce that Equation (\ref{eq:secondorderODE2}) is equivalent to the Stokes' Law.
\end{proof}

\begin{remark}
In the classical symplectic setting, the particular case $f(q)=\frac{\lambda}{2}q$ in Equation (\ref{eq:Hameqs}), with $\lambda>0$, gives rise to the dynamics of a rectilinear motion with constant acceleration (of $\frac{\lambda}{2}$). It is, for instance, the model for the free fall of a particle subject to a one-dimensional constant gravity field. Notice that there is no loss of energy of the system.
\end{remark}

\subsubsection{Description of the dynamics}

From the point of view of dynamical systems, the phase portrait in the $(q,p)$-plane of system (\ref{eq:Hameqs}) is highly similar to the phase portrait of the standard system (\ref{eq:Hameqs2}), since the vectors $(\Dot{q},\Dot{p})$ of the two systems are proportional by a $p$ factor at each point. The main difference between both systems is found at the horizontal axis $p=0$. There, the orbits that where crossing transversally in the classical system (\ref{eq:Hameqs2}) are "broken" and new punctual orbits appear in the twisted $b$-symplectic system (\ref{eq:Hameqs}). Besides, orbits in the lower plane $p<0$ change direction in system (\ref{eq:Hameqs}) (see Figure \ref{fig:phasespace} for the phase space representation of both systems).

\begin{figure}[ht!]
\centering
\begin{tikzpicture}
    \begin{axis}[axis lines = middle,xlabel = $q$,ylabel = $p$,ticks=none,xmin=-3.5,xmax=3.5,ymin=-2,ymax=2]
    \foreach \i in {-3,...,7}{
    \addplot[domain=2:1,samples=20,color=blue,->]({-x^2+\i},{x});
    \addplot[domain=1:-1,samples=20,color=blue,->]({-x^2+\i},{x});
    \addplot[domain=-1:-2,samples=20,color=blue]({-x^2+\i},{x});}
    \end{axis}
\end{tikzpicture}
\hspace{20pt}
\begin{tikzpicture}
    \begin{axis}[axis lines = middle,xlabel = $q$,ylabel = $p$,ticks=none,xmin=-3.5,xmax=3.5,ymin=-2,ymax=2]
    \foreach \i in {-3,...,7}{
    \addplot[domain=2:1,samples=20,color=blue,->]({-x^2+\i},{x});
    \addplot[domain=1:-1,samples=20,color=blue]({-x^2+\i},{x});
    \addplot[domain=-2:-1,samples=20,color=blue,->]({-x^2+\i},{x});
    \addplot[color=blue,mark=*,fill=white] coordinates{(\i,0)};}
    \end{axis}
\end{tikzpicture}
\caption{Some orbits in the phase spaces of system (\ref{eq:Hameqs2}) on the left and system (\ref{eq:Hameqs}) on the right for $f(q)=\frac{\lambda}{2}q$.}
\label{fig:phasespace}
\end{figure}

The dynamical evolution of a physical system driven by the Hamiltonian $H(q,p)=\frac{1}{2}p^2+\frac{\lambda}{2}q$ and the standard symplectic form $\omega=dp\wedge dq$ is really different from the dynamical evolution of a physical system governed by the same Hamiltonian but taking the twisted $b$-symplectic form $\omega=\frac{1}{p}dp\wedge dq$.

In the standard case, orbits are parabolas of the form $q=-p^2+c$, with $c$ a constant, everywhere (see the phase portrait on the left of Figure \ref{fig:phasespace}). The trajectory of a particle in this system is unbounded and, for any initial conditions, $q,p \to_{t\to\infty} -\infty$. This is the model of a massive particle moving in an infinite one-dimensional well, subject to a constant force field and with no friction.

In the twisted $b$-symplectic the orbits are of two types. On the one hand, there are fixed points for $p=0$ and any $q$. On the other hand, there are half-parabolas of the same form $q=-p^2+c$ at each side of the horizontal axis $p=0$. The evolution of a particle starting either in the upper or in the lower plane is similar: in both cases it will approach asymptotically $p=0$ and a fixed $q=c$ greater than the initial $q$. Nevertheless, at a finite time, a particle will be found at $p=0$ if and only if it already started there. This has physical sense, since the force is acting proportionally to the speed of the particle and in the opposite direction. Then, a particle with non-zero initial velocity will be permanently slowed down, but it will never completely stop because the acting force will also decrease in correspondence.

The nature of the trajectories in both systems is also very different. In Figure \ref{fig:trajectories2} we can see some trajectories corresponding to both the classical Hamilton's equations and the twisted $b$-symplectic Hamilton's equations coming from the same Hamiltonian $H(q,p)=\frac{1}{2}p^2+\frac{\lambda}{2}q$.

\begin{figure}[ht!]
\centering
\begin{tikzpicture}
    \begin{axis}[axis lines = middle,xlabel = $t$,ylabel = $q$,ticks=none,xmin=-0.1,xmax=3,ymin=-5,ymax=5]
    \addplot[domain=0:3,samples=20,color=blue]({x},{-x^2+2});
    \addplot[domain=0:3,samples=20,color=red]({x},{-x^2+x+3.5});
    \addplot[domain=0:3,samples=20,color=green]({x},{-x^2+2*x+0.5});
    \addplot[domain=0:3,samples=20,color=violet]({x},{-x^2-2*x-1.5});
    \addplot[domain=0:3,samples=20,color=orange]({x},{-x^2+4*x-3.5});
    \end{axis}
\end{tikzpicture}
\hspace{20pt}
\begin{tikzpicture}
    \begin{axis}[axis lines = middle,xlabel = $t$,ylabel = $q$,ticks=none,xmin=-0.2,xmax=6,ymin=-5,ymax=5]
    \addplot[domain=0:6,samples=30,color=blue]({x},{-4*exp(-x)+1});
    \addplot[domain=0:6,samples=30,color=red]({x},{-2*exp(-x)+1});
    \addplot[domain=0:6,samples=30,color=green]({x},{-1*exp(-x)+2});
    \addplot[domain=0:6,samples=30,color=violet]({x},{-2*exp(-x)-2});
    \addplot[domain=0:6,samples=30,color=orange]({x},{-3*exp(-x)+3});
    \end{axis}
\end{tikzpicture}
\begin{tikzpicture}
    \begin{axis}[axis lines = middle,xlabel = $t$,ylabel = $p$,ticks=none,xmin=-0.1,xmax=3,ymin=-5,ymax=5]
    \addplot[domain=0:3,samples=20,color=blue]({x},{-2*x});
    \addplot[domain=0:3,samples=20,color=red]({x},{-2*x+1});
    \addplot[domain=0:3,samples=20,color=green]({x},{-2*x+2});
    \addplot[domain=0:3,samples=20,color=violet]({x},{-2*x-2});
    \addplot[domain=0:3,samples=20,color=orange]({x},{-2*x+4});
    \end{axis}
\end{tikzpicture}
\hspace{20pt}
\begin{tikzpicture}
    \begin{axis}[axis lines = middle,xlabel = $t$,ylabel = $p$,ticks=none,xmin=-0.1,xmax=5,ymin=-5,ymax=5]
    \addplot[domain=0:5,samples=30,color=blue]({x},{+4*exp(-x)});
    \addplot[domain=0:5,samples=30,color=red]({x},{+2*exp(-x)});
    \addplot[domain=0:5,samples=30,color=green]({x},{-1*exp(-x)});
    \addplot[domain=0:5,samples=30,color=violet]({x},{-2*exp(-x)});
    \addplot[domain=0:5,samples=30,color=orange]({x},{-3*exp(-x)});
    \end{axis}
\end{tikzpicture}
\caption{On the left, some trajectories $q(t)$ and $p(t)$ given by the classical Hamilton's Equations (\ref{eq:Hameqs2}). On the right, some trajectories $q(t)$ and $p(t)$ given by the twisted $b$-symplectic Hamilton's Equations (\ref{eq:Hameqs}). Both are for a potential $f(q)=\frac{\lambda}{2}q$.}
\label{fig:trajectories2}
\end{figure}

The trajectory $q(t)$ of a particle under the classical model of system (\ref{eq:Hameqs2}) is of the form $q(t)=-\frac{\lambda}{4}t^2+c_1t+c_0$. It depends on the constants $c_0,c_1$ (equivalently, on the starting $q$ and $p$) but, for any initial conditions, $q(t),p(t) \to_{t\to\infty} -\infty$. This corresponds to the aforementioned one-dimensional "free fall" of a particle in a constant force field.

On the other hand, the trajectory $q(t)$ of a particle under the twisted model of system (\ref{eq:Hameqs}) is of the form $q(t)=d_0-\frac{d_1^2}{\lambda}e^{-\lambda t}$. Hence, the particle's trajectory is bounded and has a limit at a fixed $q=d_0$ greater or equal than the initial $q(0)$, no matter which initial conditions are chosen.

The orbits that "break" at the horizontal axis can be identified with "escape orbits" of a $b$-symplectic manifold. In the recent article \cite{MirandaOmsPeralta-Salas22} the existence of escape orbits is investigated connected to the singular Weinstein conjecture for singular contact manifolds which had been conjectured in \cite{MirandaOms21}. The phenomena we see in the twisted model is the even-dimensional analogue of escape orbits in \cite{MirandaOmsPeralta-Salas22} as there the system is the induced system on a level set of the Hamiltonian (which is an odd-dimensional manifold endowed with a singular contact structure, see Lemma 8.6 and Proposition 8.8 in \cite{MirandaOms18}). The outstanding fact is that in the twisted model there exists an infinite number of escape orbits going to the critical set, so Lemma 8.6 and Proposition 8.8 in \cite{MirandaOms18} provide a machinery to produce examples of (singular) Reeb vector fields with infinite singular periodic orbits. In \cite{MirandaOmsPeralta-Salas22} the existence of a lower bound is established for compact $3$-dimensional manifolds.

What we have observed is exceptional because friction is a non-conservative force and, while it cannot be described by the usual basic Hamiltonian setup, it can be described using the twisted $b$-symplectic setting. In this setting, the critical hyper-surface in our example corresponds to zero velocity or momentum. This is physically consistent with the fact that viscous friction alone cannot bring a particle to zero velocity in finite time.

\subsection{The higher-dimensional linear case}

We have introduced the one-dimensional model of the Stokes' Law using the twisted $b$-symplectic setting and now it is natural to consider higher-dimensional models. The most direct generalization is to extend the particle's Hamiltonian to $T^*\mathbb{R}^n$ in the following way:

$$H(q_1,\dots,q_n,p_1,\dots,p_n)=\frac{1}{2}\sum_{i=1}^np_i^2+\frac{\lambda}{2}q_1.$$

The dynamics governed by this Hamiltonian together with the twisted $b$-symplectic form \begin{equation}
    \omega=c\frac{1}{p_1}dp_1\wedge dq_1+\sum_{i=2}^n dp_i\wedge dq_i
    \label{eq:twistedbsymplecticformgeneral}
\end{equation}
are the following. In the direction of $q_1$, the particle behaves by the Stokes' Law: it suffers dissipation and the corresponding velocity component tends to zero. In the other directions, the motion corresponds to that of a free particle. As a consequence, the evolution of the trajectory is a curve that starts with some initial direction given by a velocity $(v_1,\dots,v_n)$ in $\mathbb{R}^n$ and tends to a motion restricted to the direction $(0,v_2,\dots,v_n)$. However, this first generalization does not allow to consider friction in all directions. We shall see in Section \ref{section:timedependentmodels} how to tackle this problem.

The modular weight $c$ appearing in the twisted $b$-symplectic form \ref{eq:twistedbsymplecticformgeneral} is giving a measure of the predominance of the singular term over the regular terms. In a way, the modular weight is also a measure of the relative importance of the direction in which there is dissipative friction with respect to the rest of directions.

\begin{definition}[Reynolds number]
In the fluid context, the Reynolds number is the ratio of the inertial force to the viscous force. It is defined as
$$Re=\frac{\rho v d}{\mu},$$
where $\rho$ is the density of the fluid, $v$ the velocity of the fluid, $d$ the diameter or characteristic length of the system and $\mu$ the dynamic viscosity of the fluid.
\end{definition}

The Reynolds number quantifies the relative importance of advective nonlinearity and viscosity, with lower Reynolds number meaning viscous forces are dominant. In practice, it is used to determine whether a fluid exhibits laminar or turbulent flow. The Stokes' Law (\ref{eq:StokesLaw}) is obtained by solving the axisymmetric and stationary incompressible Navier–Stokes equations disregarding the nonlinear term. Accordingly, the Stokes' Law describes a fluid flow with an spherical object in the laminar regime, strictly speaking when the Reynolds number is $0$, although it is a good approximation when the Reynolds number is small enough. Analogously, we have the following:

\begin{remark}
The modular weight $c$ gives a measure of the importance of the dissipative direction compared to the other directions in which there is free motion. Then, it can be associated with an analogue of the Reynolds number $Re$ when $Re\approx 0$. A high modular weight $c$ implies that there is a big influence of the dissipation by viscosity in the overall system, which is equivalent to a low $Re$.
\end{remark}

\subsection{The pure quadratic potential}
Consider now a quadratic potential of the type $f(q)=\frac{\lambda}{4}q^2$. The dynamics of a physical system driven by the Hamiltonian $H(q,p)=\frac{1}{2}p^2+\frac{\lambda}{4}q^2$ and the standard symplectic form $\omega=dp\wedge dq$ corresponds to a simple harmonic oscillator. Explicitly, the Hamilton's equations in this case are:

\begin{equation}
    \begin{cases}
    \Dot{q}=p\\
    \Dot{p}=-\frac{\lambda}{2}q
    \end{cases}.
\label{eq:HameqsSHM}
\end{equation}

Orbits in the phase space of system (\ref{eq:HameqsSHM}) are circles everywhere except from a fixed point at the origin (see the phase portrait on the left of Figure \ref{fig:phasespacequad}). They are of the form $p^2+\frac{\lambda}{2}q^2=c$, with $c$ a constant. The position of a particle in this system is bounded and so is its velocity for any initial conditions, since $q(t)$ and $p(t)$ are sine waves. This behaviour corresponds exactly to the classical model of a simple harmonic oscillator, which is natural for a quadratic potential.

However, and more interestingly, the same Hamiltonian together with the twisted $b$-symplectic form $\omega=\frac{1}{p}dp\wedge dq$ gives another dynamics. In particular, we obtain the following equations of motion:
\begin{equation}
    \begin{cases}
    \Dot{q}=p^2\\
    \Dot{p}=-\frac{\lambda}{2}pq
    \end{cases}.
\label{eq:Hameqsquad}
\end{equation}

On the right of Figure \ref{fig:phasespacequad} we can see the phase space representation of the orbits of system (\ref{eq:Hameqsquad}) and on the left of Figure \ref{fig:trajectories3} we can see some trajectories $q(t)$ and $p(t)$ of the same system.

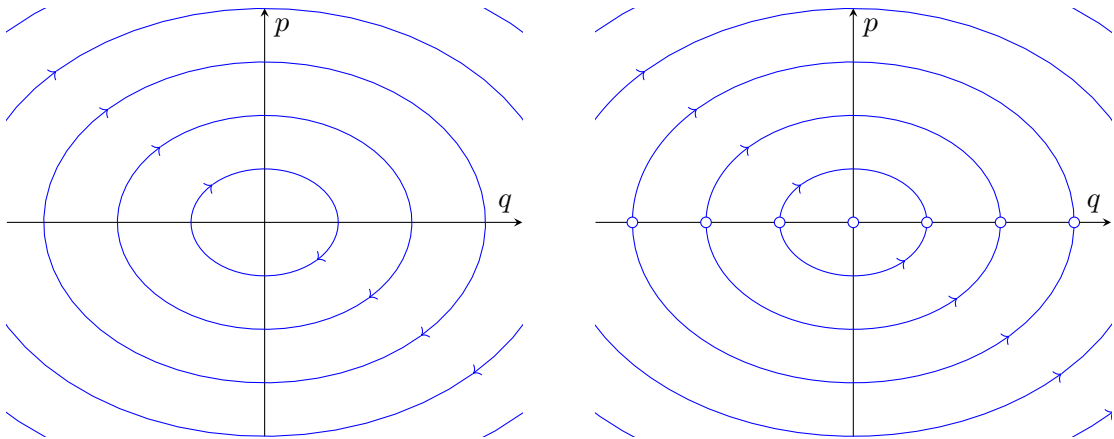
\begin{figure}[ht!]
\centering
\begin{tikzpicture}
    \begin{axis}[axis lines = middle,xlabel = $q$,ylabel = $p$,ticks=none,xmin=-3.5,xmax=3.5,ymin=-2,ymax=2]
    \foreach \i in {1,...,7}{
    \addplot[domain=135:-45,samples=30,color=blue,->]({\i*cos(x)},{\i/2*sin(x)});
    \addplot[domain=315:135,samples=30,color=blue,->]({\i*cos(x)},{\i/2*sin(x)});}
    \end{axis}
\end{tikzpicture}
\hspace{20pt}
\begin{tikzpicture}
    \begin{axis}[axis lines = middle,xlabel = $q$,ylabel = $p$,ticks=none,xmin=-3.5,xmax=3.5,ymin=-2,ymax=2]
    \foreach \i in {1,...,7}{
    \addplot[domain=180:135,samples=30,color=blue,->]({\i*cos(x)},{\i/2*sin(x)});
    \addplot[domain=-45:135,samples=30,color=blue]({\i*cos(x)},{\i/2*sin(x)});
    \addplot[domain=180:315,samples=30,color=blue,->]({\i*cos(x)},{\i/2*sin(x)});
    \addplot[color=blue,mark=*,fill=white] coordinates{(\i-4,0)};}
    \end{axis}
\end{tikzpicture}
\caption{Some orbits in the phase spaces of system (\ref{eq:HameqsSHM}) on the left and system (\ref{eq:Hameqsquad}) on the right for $f(q)=\frac{\lambda}{4}q^2$.}
\label{fig:phasespacequad}
\end{figure}

The second order ODE equivalent to system (\ref{eq:Hameqsquad}) is:
\begin{equation}
    \ddot{q}=-\lambda\dot{q}q.
    \label{eq:secondorderODEquad}
\end{equation}

Equation (\ref{eq:secondorderODEquad}) is a highly non-linear equation which has the following solution for the trajectory $q(t)$:
$$q(t)=\frac{c_1}{\sqrt{\lambda}}\tanh{\left(\frac{c_1\sqrt{\lambda}}{2}t+c_2\right)},$$
with $c_1$ and $c_2$ depending on the initial conditions. On the right of Figure \ref{fig:trajectories3}, we can see some trajectories $q(t)$ and $p(t)$ for different values of $c_1$ and $c_2$.

Again, as in the linear case, $p(t) \to_{t\to\pm\infty} 0$. On the other hand, the position $q(t)$ of a particle under this potential is bounded on the range $(-c_1,c_1)$, which makes the dynamics really different from the linear potential case.

In this setting, we can think that the particle is really enclosed in a uni-dimensional container and goes from one end to the other as time passes. It does so by starting to separate slowly from one border, then accelerating fast to pass over the mid space of the container, and then slowing again when arriving to the other border.

The quadratic potential, then, models a particle crossing the interior of a box at a slow speed when it is near each edge and at a high speed in the middle. Observe that the orbits in the twisted model "break" again like in the linear twisted case, allowing an infinite number of "escape orbits". This example has a bonus, as the escape orbits here correspond exactly to genuine \emph{singular periodic orbits} as the ones described in \cite{MirandaOms21}. These singular periodic orbits are indeed the union of $4$ different trajectories: two symmetric hetero-clinic half-circles and the two fixed points on the horizontal axis at their ends.

\begin{figure}[ht!]
\centering
\begin{tikzpicture}
    \begin{axis}[axis lines = middle,xlabel = $t$,ylabel = $q$,ticks=none,xmin=-0.1,xmax=3,ymin=-5,ymax=5]
    \addplot[domain=0:3,samples=60,color=blue]({x},{+1*sin(x*540/pi)+1*cos(x*540/pi)});
    \addplot[domain=0:3,samples=60,color=red]({x},{+2*sin(x*540/pi)-2*cos(x*540/pi)});
    \addplot[domain=0:3,samples=60,color=green]({x},{-1*sin(x*540/pi)+0*cos(x*540/pi)});
    \addplot[domain=0:3,samples=60,color=violet]({x},{-3*sin(x*540/pi)-1*cos(x*540/pi)});
    \addplot[domain=0:3,samples=60,color=orange]({x},{-1*sin(x*540/pi)+2*cos(x*540/pi)});
    \end{axis}
\end{tikzpicture}
\hspace{20pt}
\begin{tikzpicture}
    \begin{axis}[axis lines = middle,xlabel = $t$,ylabel = $q$,ticks=none,xmin=-0.1,xmax=3,ymin=-5,ymax=5]
    \addplot[domain=0:3,samples=50,color=blue]({x},{3/sqrt(11)*tanh(3*sqrt(11)/2*x-12)});
    \addplot[domain=0:3,samples=50,color=red]({x},{5/sqrt(5)*tanh(5*sqrt(5)/2*x-10)});
    \addplot[domain=0:3,samples=50,color=green]({x},{2/sqrt(2)*tanh(2*sqrt(2)/2*x-2)});
    \addplot[domain=0:3,samples=50,color=violet]({x},{3/sqrt(4)*tanh(3*sqrt(4)/2*x-3)});
    \addplot[domain=0:1,samples=100,color=orange]({x},{12/sqrt(44)*tanh(12*sqrt(44)/2*x-12)});\addplot[domain=1:3,samples=10,color=orange]({x},{12/sqrt(44)*tanh(12*sqrt(44)/2*x-12)});
    \end{axis}
\end{tikzpicture}
\begin{tikzpicture}
    \begin{axis}[axis lines = middle,xlabel = $t$,ylabel = $p$,ticks=none,xmin=-0.1,xmax=3,ymin=-5,ymax=5]
    \addplot[domain=0:3,samples=60,color=blue]({x},{+1*cos(x*540/pi)-1*sin(x*540/pi)});
    \addplot[domain=0:3,samples=60,color=red]({x},{+2*cos(x*540/pi)+2*sin(x*540/pi)});
    \addplot[domain=0:3,samples=60,color=green]({x},{-1*cos(x*540/pi)-0*sin(x*540/pi)});
    \addplot[domain=0:3,samples=60,color=violet]({x},{-3*cos(x*540/pi)+1*sin(x*540/pi)});
    \addplot[domain=0:3,samples=60,color=orange]({x},{-1*cos(x*540/pi)-2*sin(x*540/pi)});
    \end{axis}
\end{tikzpicture}
\hspace{20pt}
\begin{tikzpicture}
    \begin{axis}[axis lines = middle,xlabel = $t$,ylabel = $p$,ticks=none,xmin=-0.1,xmax=3,ymin=-0,ymax=10]
    \addplot[domain=0:3,samples=100,color=blue]({x},{3^2/2/(cosh(3*sqrt(11)/2*x-12))^2/2});
    \addplot[domain=0:3,samples=100,color=red]({x},{5^2/2/(cosh(5*sqrt(5)/2*x-10))^2/2});
    \addplot[domain=0:3,samples=100,color=green]({x},{2^2/2/(cosh(2*sqrt(2)/2*x-2))^2/2});
    \addplot[domain=0:3,samples=100,color=violet]({x},{3^2/2/(cosh(3*sqrt(4)/2*x-3))^2/2});
    \addplot[domain=0:1,samples=100,color=orange]({x},{6^2/2/(cosh(12*sqrt(44)/2*x-12))^2/2});\addplot[domain=1:3,samples=100,color=orange]({x},{12^2/2/(cosh(12*sqrt(44)/2*x-12))^2/2});
    \end{axis}
\end{tikzpicture}
\caption{On the left, some trajectories $q(t)$ and $p(t)$ given by the classical Hamilton's Equations (\ref{eq:HameqsSHM}). On the right, some trajectories $q(t)$ and $p(t)$ given by the twisted $b$-symplectic Hamilton's Equations (\ref{eq:Hameqsquad}). Both are for a potential $f(q)=\frac{\lambda}{4}q^2$.}
\label{fig:trajectories3}
\end{figure}
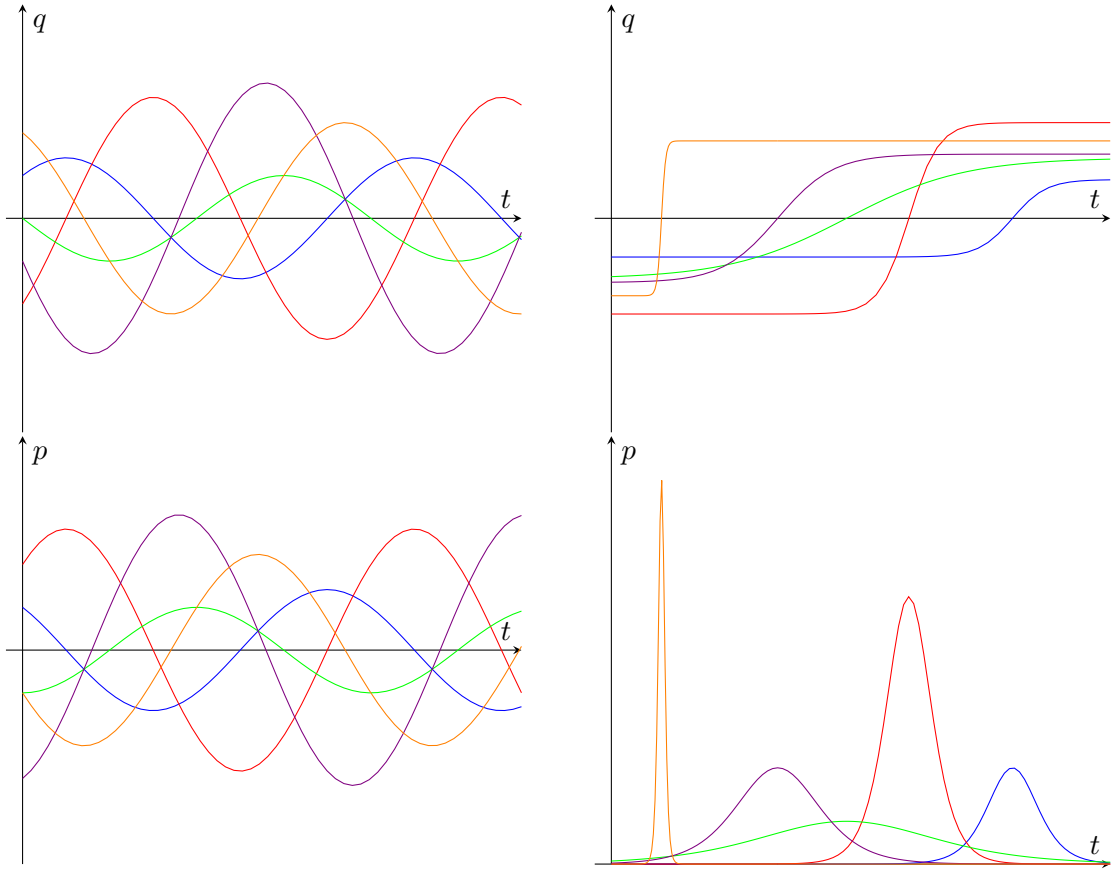

\subsection{The general quadratic potential}

It is natural to consider the coupling of the pure quadratic potential with the linear potential studied previously. Consider a physical particle moving in a viscous fluid and obeying the Stokes Law. Suppose that the fluid has a non-uniform viscosity $\eta$, which is, for instance, the case whenever there is a gradient of temperature (as the viscosity usually depends on the temperature). For small fluctuations, up to first order in position $q$, the viscosity can be written as $\eta=\eta_0(1+\alpha q)$. Therefore, the potential $f(q)$ accounting for the drag coefficient becomes $f(q)=\lambda(1+\alpha q)$ and the associated Hamiltonian is $$H(p,q)=\frac{p^2}{2}+\frac{\lambda}{2}q\left(1+\alpha \frac{q}{2}\right).$$

The corresponding second order ODE is
$$\ddot{q}=-\lambda(1+\alpha q) \dot{q},$$
and includes both the linear regime (which is expected to be dominant) and the quadratic regime as a perturbation. It is the most natural generalization of the linear regime from the physical point of view.

Another option is to consider dissipation in one direction of motion in space configurations which are not $\mathbb{R}^n$, such as a particle moving over a cylinder $S^1\times \mathbb{R}$. In this case, if the potential is linear (with respect to the axial coordinate), the dynamics is the one depicted in Figure \ref{fig:trajectoriescylinder}. There, the trajectory of a particle under the Hamiltonian $$H(\theta,q,p_\theta,p_q)=\frac{p_\theta^2+p_q^2}{2}+\frac{\lambda}{2}q,$$
together with the twisted $b$-symplectic form $\omega=c\frac{1}{p_q}dp_q\wedge dq+dp_\theta\wedge d\theta$ tends to a periodic orbit around the cylinder. In the cylinder $S^1\times \mathbb{R}$ it also makes sense to consider a potential which is function of the angular coordinate, which is what we do next.

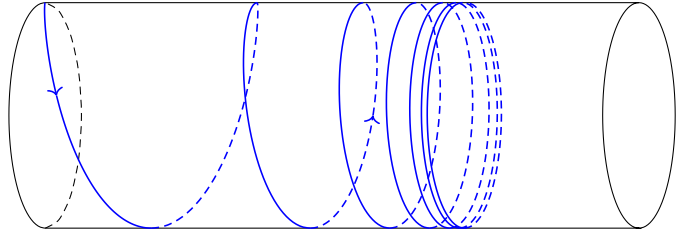
\begin{figure}[ht!]
\centering
\begin{tikzpicture}[x=15mm,y=cos(20)*15mm,z={(0,-sin(20)*15mm)},rotate=270,transform shape]
\def\cylrad{1}
\def\cylht{8}
\draw
    (-\cylrad, \cylht*1.2) -- (-\cylrad, \cylht/2) --
    plot[smooth, samples=25, variable=\t, domain=180:360]
      ({cos(\t)*\cylrad}, \cylht/2, {-sin(\t)*\cylrad}) --
    (\cylrad, \cylht*1.2)
    plot[smooth cycle, samples=51, variable=\t, domain=0:360]
      ({cos(\t)*\cylrad}, \cylht*1.2, {-sin(\t)*\cylrad});
\draw[densely dashed]
    plot[smooth, samples=9, variable=\t, domain=0:180]
      ({cos(\t)*\cylrad}, \cylht/2, {-sin(\t)*\cylrad});
\draw[semithick,color=blue,->]
    plot[smooth, samples=25, variable=\t, domain=180:260]
        ({cos(\t)*\cylrad}, {\cylht*(1-1/(2^(1))) + (\t-180)*\cylht/720/(2^(1))}, {-sin(\t)*\cylrad});
\draw[semithick,color=blue]
    plot[smooth, samples=25, variable=\t, domain=260:360]
        ({cos(\t)*\cylrad}, {\cylht*(1-1/(2^(1))) + (\t-180)*\cylht/720/(2^(1))}, {-sin(\t)*\cylrad});
\draw[semithick,color=blue]
    \foreach \y in {2, 3, 4, 5, 6, 7} {
      plot[smooth, samples=25, variable=\t, domain=180:360]
        ({cos(\t)*\cylrad}, {\cylht*(1-1/(2^(\y))) + (\t-180)*\cylht/720/(2^(\y))}, {-sin(\t)*\cylrad})};
\draw[semithick, densely dashed,color=blue]
    \foreach \y in {1, 3, 4, 5, 6, 7} {
      plot[smooth, samples=25, variable=\t, domain=0:180]
        ({cos(\t)*\cylrad}, {\cylht*(2^(\y+2)-3)/(2^(\y+2)) + \t*\cylht/720/(2^(\y))}, {-sin(\t)*\cylrad})};
\draw[semithick, densely dashed,color=blue,->]
      plot[smooth, samples=25, variable=\t, domain=0:90]
        ({cos(\t)*\cylrad}, {\cylht*(2^(2+2)-3)/(2^(2+2)) + \t*\cylht/720/(2^(2))}, {-sin(\t)*\cylrad});
\draw[semithick, densely dashed,color=blue]
      plot[smooth, samples=25, variable=\t, domain=90:180]
        ({cos(\t)*\cylrad}, {\cylht*(2^(2+2)-3)/(2^(2+2)) + \t*\cylht/720/(2^(2))}, {-sin(\t)*\cylrad});  
\end{tikzpicture}
\begin{tikzpicture}[x=10mm,y=10mm]
\node[rectangle,draw=white,minimum width = 10cm,minimum height = 0.3cm] (r) at (0,0) {};
\end{tikzpicture}
\caption{A dissipating trajectory in the cylinder in which the singularity is in the $\mathbb{R}$ component of the fiber.}
\label{fig:trajectoriescylinder}
\end{figure}

\subsection{A periodic potential}

Consider the periodic potential $f(\theta)=\frac{\lambda}{2}\cos\theta$, with $\theta\in S^1$. The phase space in this case is $T^*S^1\cong S^1\times \mathbb{R}$ and we denote by $p_\theta$ the conjugate momentum coordinate. The dynamical evolution of a physical system driven by the Hamiltonian $H(\theta,p_\theta)=\frac{1}{2}p_\theta^2+\frac{\lambda}{2}\cos\theta$ and the standard symplectic form $\omega=dp_\theta\wedge d\theta$ corresponds to the model of the simple pendulum. The Hamilton's equations of this model are:

\begin{equation}
    \begin{cases}
    \Dot{\theta}=p_\theta\\
    \Dot{p_\theta}=\frac{\lambda}{2}\sin\theta
    \end{cases}.
\label{eq:Hameqspendulum}
\end{equation}

The orbits of this system are of four types. There are two fixed points at $(0,0)$ (of stable type) and $(\pi,0)$ (of saddle type), two homoclinic orbits at the fixed point $(\pi,0)$, a $1$-parametric family of periodic orbits encircling the stable point between the two homoclinic orbits, and a $1$-parametric family of periodic orbits around the cylinder filling the rest of the space away from the homoclinic orbits (see the phase portrait on the left of Figure \ref{fig:phasespaceperiod}). The position of a particle in this system is bounded and so is its velocity, but it depends on the initial conditions whether it keeps moving periodically or stabilizes (which only happens at the fixed points or the homoclinic orbits).

The same Hamiltonian together with the twisted $b$-symplectic form $\omega=\frac{1}{p_\theta}dp_\theta\wedge d\theta$ gives a very different dynamics, which is encoded in the following equations of motion:
\begin{equation}
    \begin{cases}
    \Dot{\theta}=p_\theta^2\\
    \Dot{p_\theta}=\frac{\lambda}{2}p_\theta\sin\theta
    \end{cases}.
\label{eq:Hameqsperiod}
\end{equation}

On the right of Figure \ref{fig:phasespaceperiod} we can see the phase space representation of the orbits of this system.

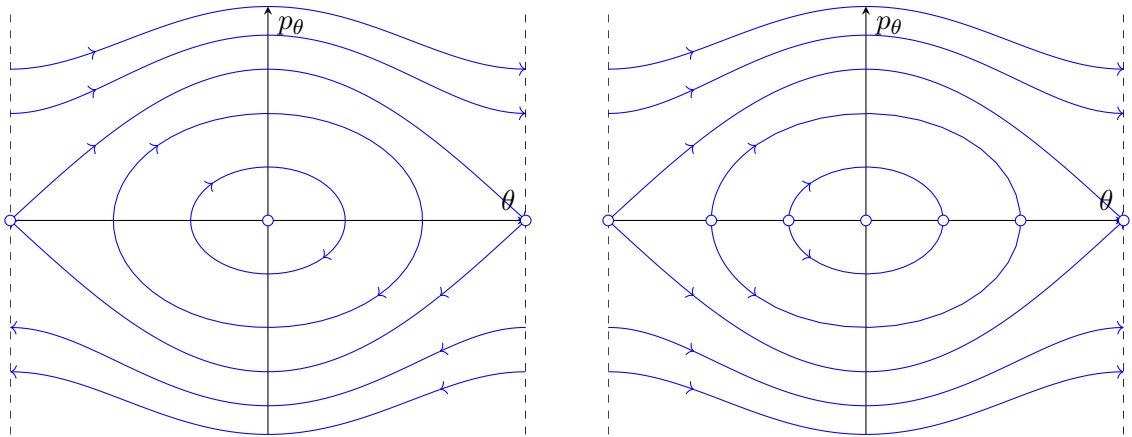
\begin{figure}[ht!]
\centering
\begin{tikzpicture}
    \begin{axis}[axis lines = middle,xlabel = $\theta$,ylabel = $p_\theta$,ticks=none,xmin=-2,xmax=2,ymin=-2,ymax=2]
    \foreach \i in {1,...,3}{
    \addplot[domain=-180:-120,samples=50,color=blue,->]({x/90},{sqrt(\i+cos(x))});
    \addplot[domain=-180:120,samples=50,color=blue,<-]({x/90},{-sqrt(\i+cos(x))});
    \addplot[color=blue,mark=*,fill=white] coordinates{(2*\i-4,0)};}
    \foreach \i in {1,...,3}{
    \addplot[domain=-120:180,samples=50,color=blue,->]({x/90},{sqrt(\i+cos(x))});
    \addplot[domain=120:180,samples=50,color=blue,<-]({x/90},{-sqrt(\i+cos(x))});}
    \foreach \i in {1,...,2}{\addplot[domain=135:-45,samples=30,color=blue,->]({\i*0.6*cos(x)},{\i/2*sin(x)});
    \addplot[domain=315:135,samples=30,color=blue,->]({\i*0.6*cos(x)},{\i/2*sin(x)});}
    \addplot [mark=none,dashed] coordinates {(-2, -2) (-2, 2)};
    \addplot [mark=none,dashed] coordinates {(2, -2) (2, 2)};
    \end{axis}
\end{tikzpicture}
\hspace{20pt}
\begin{tikzpicture}
    \begin{axis}[axis lines = middle,xlabel = $\theta$,ylabel = $p_\theta$,ticks=none,xmin=-2,xmax=2,ymin=-2,ymax=2]
    \foreach \i in {1,...,3}{
    \addplot[domain=-180:-120,samples=50,color=blue,->]({x/90},{sqrt(\i+cos(x))});
    \addplot[domain=-180:-120,samples=50,color=blue,->]({x/90},{-sqrt(\i+cos(x))});
    \addplot[color=blue,mark=*,fill=white] coordinates{(2*\i-4,0)};}
    \foreach \i in {1,...,3}{
    \addplot[domain=-120:180,samples=50,color=blue,->]({x/90},{sqrt(\i+cos(x))});
    \addplot[domain=-120:180,samples=50,color=blue,->]({x/90},{-sqrt(\i+cos(x))});}
    \foreach \i in {1,...,2}{
    \addplot[domain=180:135,samples=30,color=blue,->]({\i*0.6*cos(x)},{\i/2*sin(x)});
    \addplot[domain=-135:135,samples=30,color=blue]({\i*0.6*cos(x)},{\i/2*sin(x)});
    \addplot[domain=180:225,samples=30,color=blue,->]({\i*0.6*cos(x)},{\i/2*sin(x)});}
    \addplot [mark=none,dashed] coordinates {(-2, -2) (-2, 2)};
    \addplot [mark=none,dashed] coordinates {(2, -2) (2, 2)};
    \addplot[color=blue,mark=*,fill=white] coordinates{(-1.2,0)};
    \addplot[color=blue,mark=*,fill=white] coordinates{(-0.6,0)};
    \addplot[color=blue,mark=*,fill=white] coordinates{(1.2,0)};
    \addplot[color=blue,mark=*,fill=white] coordinates{(0.6,0)};
    \end{axis}
\end{tikzpicture}
\caption{Some orbits in the phase spaces of system (\ref{eq:Hameqspendulum}) on the left and system (\ref{eq:Hameqsperiod}) on the right for $f(q)=\frac{\lambda}{2}\cos\theta$.}
\label{fig:phasespaceperiod}
\end{figure}

The equivalent second order ODE is:
\begin{equation}
    \ddot{\theta}=\lambda\dot{\theta}\sin\theta.
    \label{eq:secondorderODEperiod}
\end{equation}

We observe that, differently from the twisted $b$-symplectic models studied before, in this one there are still periodic orbits, namely, the $1$-parametric family of periodic orbits that fill the two half-spaces away from the two homoclinic orbits. On the other hand, the dynamics inside the region enclosed by the homoclinic orbits is the same dynamics that we obtained for the quadratic potential.

\subsection{General dynamics of the twisted $b$-symplectic model}

With the previous illustrative examples in mind, the interpretation of the twisted $b$-symplectic model in the general case is straightforward. Recall that the Hamilton's equations derived from the Hamiltonian $H=\frac{p^2}{2}+f(q)$ and the twisted $b$-symplectic form
$\omega=\frac{1}{p}dp\wedge dq$ are:
\begin{equation*}
    \begin{cases}
    \Dot{q}=p^2\\
    \Dot{p}=-p\frac{\partial f}{\partial q}
    \end{cases}.
    \label{eq:twistedbhamvf}
\end{equation*}

If $(q,p)$ are assumed to be the coordinates of the phase space of a mechanical system, the behaviour of a particle under $H$ is clearly conditioned by the singularity of the system at $p=0$.

If a particle starts at any $p\neq 0$ and follows a trajectory that tends to $p=0$, the most direct physical interpretation of the model is that of a decelerating motion, for instance the one encountered in a dissipative system. If it escapes from $p=0$, it can be interpreted just in the same way but with time reversed.

The implications of having the singularity at the fibers of the cotangent bundle extend further than it seems at first glance. The singularity determines an unreachable location in the fiber, i.e., that zero momentum is unreachable. But the momentum of the particle will tend there (or escape from there) for many different initial conditions. As a consequence, the position of the particle is also indirectly conditioned by the singularity, since tending to zero momentum will cause the position to stabilize.

The twisted $b$-symplectic model with the singular fiber at zero is, then, a physical model that can explain systems in which velocity decays and so does the change in position of the particle as a consequence.

\begin{remark}
For a Hamiltonian of the type $H=\frac{p^2}{2}+f(q)$, the relation between the twisted $b$-Hamiltonian vector field $X^b_H$ (given by Equation (\ref{eq:Hameqs})) and the usual Hamiltonian vector field $X_H$ (given by Equation (\ref{eq:Hameqs2})) is $X^b_H = pX_H$. Accordingly, the orbits of $X^b_H$ coincide with those of $X_H$ away from the critical line $\{p = 0\}$ (which is now filled by a set of stagnation points of $X^b_H$) and up to a reversed time parametrization for $p<0$. As a consequence, any point $(q,p=0)$ which is not a critical point of potential $f(q)$ yields an escape orbit of $X^b_H$ because the corresponding orbit of $X^H$ is regular and transverse to the line $p = 0$.
\end{remark}

\subsection{Geometric considerations}
One could ask whether the dynamics of the model can be obtained as the cotangent lift of an action or, more precisely, as a $b$-cotangent lift (see Sections \ref{subsec:bcotangentlifts} and \ref{subsec:cotangentmodels}). As we show next using geometric arguments, the answer to this question is negative.

\begin{lemma}
The dynamics of the model presented in Section \ref{section:bsymplecticmodel} can not be obtained as the cotangent lift in Proposition \ref{prop:cotangentlifthamiltonian}. Thus, it can not be a cotangent model as in Theorem \ref{thm:cotangentliftALM}.
\end{lemma}

\begin{proof}
In the lowest dimensional case, suppose that there is an an action $\rho:G\times \mathbb{R}\to \mathbb{R}$ with cotangent lift $\hat\rho:G\times T^*\mathbb{R}\to T^*\mathbb{R}$ and with moment map $H=\frac{p^2}{2}+f(q)$. By definition, $\hat{\rho}_g(q,p)$ restricted to the base $\mathbb{R}$ of $T^*\mathbb{R}\ni (q,p)$ is equal to $\rho_g(q)$ for any $g\in G,q\in\mathbb{R}$. Then, the restriction does not depend on the fiber coordinate $p$ and the same happens for the $\frac{\partial}{\partial q}$ component of the infinitesimal generator of $\hat{\rho}_g(q,p)$. But this infinitesimal generator has to be of the form $$X=p^2\frac{\partial}{\partial q}-\frac{\partial f(q)}{\partial q}p\frac{\partial}{\partial p}$$ in order to satisfy $\iota_X\omega=-dH$, where $\omega$ is the twisted $b$-symplectic form $\frac{1}{p}dp\wedge dq$. And the term $p^2\frac{\partial}{\partial q}$ depends on the fiber coordinate $p$, which is a contradiction.
\end{proof}

Another way to see it is the following. By Proposition \ref{prop:cotangentlifthamiltonian}, the twisted $b$-cotangent lift $\hat\rho$ is $b$-Hamiltonian and its moment map $\mu:T^*M \to \mathfrak{g}^*$ contains a logarithm term associated to the toric component of the action $\rho$. In Proposition 26 of \cite{BraddellKiesenhoferMiranda18}, it is proved that the action of $G$ on the mapping torus $Z$ always lifts to an action of a product group $S^1 \times H$ on a finite trivializing cover of $Z$, where $H$ is compact and connected. $G$ is necessarily of the form $G = (S^1 \times H)/\Gamma$ for a finite cyclic subgroup $\Gamma$. Hence, the moment map of the lifted action of a group action with non-vanishing modular weight \cite{BraddellKiesenhoferMiranda18,MatveevaMiranda22} has to include a term of the form $\mu=c\log|p|$, which is not compatible with the Hamiltonian $H=\frac{p^2}{2}+f(q)$. In higher dimensions, the model can not be a cotangent lift for the same reason.

\section{Time-dependent singular models}
\label{section:timedependentmodels}

In order to generalize this friction model to multiple dimensions, the key idea is to extend the configuration space $Q$ to $Q\times\mathbb{R}$. The $\mathbb R$ component in $Q\times\mathbb{R}$  describes the real time $t$ while the dynamics inside the phase space is computed according to a curvilinear time $s$. This is conceptually the idea of the well-known method of characteristics in PDEs. After computing the solution, one only needs to project the trajectory on the space $Q$ and read the time on the real axis. We require $\dot{t}>0$ to be consistent and we denote by $q$ the position in $Q$ and by $p$ the associated momentum. We also denote by $E$ the conjugated variable associated with $t$, since the energy is the natural conjugate of time in physics. Using the results of Section 3, where we have seen how to introduce dissipation in one dimension thanks to a $b$-symplectic form, our goal is to include the dissipation in this new energy variable. Therefore, the non-dissipative dynamics will proceed classically, with an energy which is dissipated through time.

To start, consider the Hamiltonian
\begin{equation}
H(p,q,t,E)=\frac{p^2}{2}+V(q,t)-E.
\end{equation}
Assuming $E$ is the energy of the system, one expects the preservation of the Hamiltonian (the conservation of $H=0$) along the physical trajectory. We use the canonical symplectic form 
\begin{equation} \omega=\sum_i\mathrm{d}p_i\wedge\mathrm{d}q_i-\mathrm{d}E\wedge \mathrm{d}t.\end{equation}
The associated dynamics writes 
\begin{eqnarray}
\dot{q}_i&=p_i \quad\quad\quad \dot{p}_i&=-\frac{\partial V(q,t)}{\partial q_i} \\
\dot{t}&=1 \quad\quad\quad \dot{E}&=\frac{\partial V(q,t)}{\partial t}
\end{eqnarray}
Therefore, in this case, the curvilinear coordinate is the real time: $s=t$. The particle follows the expected dynamics with a potential that may depend on time. 

Now, to model friction, it is natural to consider adding to the Hamiltonian a factor depending on a friction coefficient $\lambda$. The friction will slow down the dynamics and thus $t$ compared with $s$. However, the potential remains associated to the real time and thus it appears accelerated compared with the curvilinear time. In order to use this effective time, we need to re-scale the Hamiltonian to deduce the suitable time re-scaling. When considering dissipative dynamics, it is natural to expect an exponential re-scaling. Indeed, a close-to-the-equilibrium relaxation mode provides a Lyapunov coefficient to control the decay of the perturbation \cite{Groot2013,Glansdorff1971}. Such re-scaling ideas have already been suggested in different contexts, see for instance \cite{Feix1985}. For our purpose, we consider the following Hamiltonian
\begin{equation}
H(p,q,t,E)=\frac{p^2}{2}+\frac{e^{2\lambda t}}{\lambda^2}V(q,t)- \frac{e^{\lambda t}}{\lambda}E,
\end{equation}
with the same canonical symplectic form. The associated dynamics writes 
\begin{eqnarray}
\dot{q}_i&=p_i \quad\quad\quad \dot{p}_i&=-\frac{e^{2\lambda t}}{\lambda^2}\frac{\partial V(q,t)}{\partial q_i}\\
\dot{t}&=\frac{e^{\lambda t}}{\lambda} \quad\quad\quad \dot{E}&=\frac{e^{2\lambda t}}{\lambda^2}\frac{\partial V(q,t)}{\partial t}+\frac{2e^{2\lambda t}}{\lambda} V(q,t)- e^{\lambda t}E
\end{eqnarray}
The two first terms describe the energy  linked with the time-dependence of the potential. The last term describes the loss of energy caused by the viscous dissipation. The equation for $t$ can be solved exactly: $t(s)=-\frac{\ln(-s)}{\lambda}$. In particular, $\mathrm{d}s=\lambda e^{-\lambda t}\mathrm{d}t$. Let us now reconstruct the particle dynamics in real time:
\begin{eqnarray}
\frac{\mathrm{d}q_i}{\mathrm{d}t}=\lambda e^{-\lambda t}\dot{q}_i=\lambda e^{-\lambda t}p_i \quad\quad\quad \frac{\mathrm{d}p_i}{\mathrm{d}t}=\lambda e^{-\lambda t}\dot{p}_i&=-\frac{e^{\lambda t}}{\lambda}\frac{\partial}{\partial q_i} V(q,t),
\end{eqnarray}
and, therefore,
\begin{equation}
\frac{\mathrm{d}^2 q_i}{\mathrm{d}t^2}=-\lambda \frac{\mathrm{d}q_i}{\mathrm{d}t}-\frac{\partial}{\partial q_i} V(q,t),
\end{equation}
which is the equation of a particle in a $n$-dimensional space with a viscous friction of coefficient $\lambda$ and in a time-dependent potential $V(q,t)$. 

The friction arises from an exponential re-scaling of time. Such a re-scaling is actually the source of a singularity and, then, singular geometry arises naturally after a change of variables from $t$ to $s$ in the symplectic form using $s(t)=e^{-\lambda t}$ and $\mathrm{d}t=-\frac{\mathrm{d}s}{\lambda s}$. For convenience, we also redefine $E_s=E/\lambda$. Then, we obtain
\begin{equation}
\omega=\sum_i\mathrm{d}q_i\wedge\mathrm{d}p_i+\frac{1}{ s}\mathrm{d}s\wedge \mathrm{d}E_s,
\end{equation}
which is the non-twisted canonical $b$-symplectic form. In these coordinates, the Hamiltonian becomes 
 \begin{equation}
H(p,q,s,E_s)=\frac{p^2}{2}+\frac{V(q,t(s))}{(\lambda s)^2}- \frac{E_s}{s},
\end{equation}
which has a singularity of higher order. Indeed, it is a $b^2$-function and not a $b$-function. Such a discrepancy between the degree of the singularity in the symplectic form and the degree of the singularity in total energy of the system is not new (see \cite{DelshamsKiesenhoferMiranda17, MirandaPlanas22} for other examples).

Summing up, the Hamiltonian is simpler in these coordinates. But the main advantage is that the intrinsic time (the curvilinear coordinate) now corresponds to the coordinate $s$. Indeed, the equations of motion now read as follows:
\begin{eqnarray}
\dot{q}_i&=p_i \quad\quad\quad \dot{p}_i&=-\frac{1}{(\lambda s)^2}\frac{\partial V(q,t(s))}{\partial q_i}\\
\dot{s}&=1 \quad\quad\quad \dot{E}&=\frac{\partial}{\partial s}\left(\frac{1}{(\lambda s)^2} V(q,t(s))\right)+\frac{E_s}{s^2}
\end{eqnarray}
The coordinate $s$ is now trivial and we may omit this dimension, leaving a standard Hamiltonian dynamics with a modified time-dependent potential. The dynamics then writes as:
\begin{equation}
    \ddot{q}_i(s)=-\frac{1}{(\lambda s)^2}\frac{\partial V(q,s)}{\partial q_i}
\end{equation}
and the real-time solution is obtained by undoing the change of variables $s(t)=e^{-\lambda t}.$ 

\begin{remark}[Connection with magnetism]

One could think about extending the singular models considered in this article to include the effects of an electromagnetic field acting on a charge. In that case, the configuration space is $\mathbb{R}^3$, with an electric potential function $\phi$ and a magnetic vector potential  $\textbf{A}$. The corresponding electric and magnetic fields is $\textbf{E}=\nabla \phi$ and $\textbf{B}=\nabla \times \textbf{A}$ respectively, both depending on the position $\textbf{q}\in \mathbb{R}^3$ of the particle. The force $\textbf{F}$ acting on the particle is the Lorentz force $\textbf{F}= e(\textbf{E} + \textbf{v} \times \textbf{B})$, a function of both position and velocity.

The problem can be studied in the Hamiltonian setting by identifying the tangent and the cotangent vectors via a fixed Riemannian metric. The magnetic field $\textbf{B}$ is associated with a 2-form $B=\iota_{\textbf{B}}\Omega$ where $\Omega$ is the volume form associated with the fixed Riemannian metric and $\iota$ is the internal product. The Maxwell equation $\nabla \cdot \textbf{B} = 0$, which allows the existence of the vector potential, becomes the condition $dB = 0$. By means of the Poincaré Lemma, there exists a $1$-form $A$ such that $B=dA$. Actually, we have the correspondence $A=\langle \textbf{A},\cdot\rangle$. The electrodynamics naturally appears in the phase space $(T^*Q, \omega_B)$, where $\omega_B$ is the sum of the canonical symplectic form $ \omega_Q$ on $T^*Q$ and the pull-back of the $2$-form $B$ by the natural projection $\pi:T^*Q\to Q$, i.e., $\omega_B=\omega_Q+\pi^*B$.

The discussion on the multidimensional case is still valid under the presence of a time-dependent magnetic potential $B=B_{ij}(q,t)\mathrm{d}q_i\wedge\mathrm{d}q_j$, and the subtlety in the case of dissipation is to adjust the speed of the magnetic field. Similarly to what has been done in this section, the recipe in this case would be given by the change $B\rightarrow \frac{e^{\lambda t}}{\lambda} B$. Nevertheless, since the magnetic field would no longer be closed, the method presented here would need further development to be convenient for magnetic dynamics.
\end{remark}

\section{Conclusions}
\label{section:conclusions}
This paper aims to provide a finite-dimensional analogy of fluid mechanics using the techniques of $b$-symplectic geometry to model dissipation in conditions of no turbulence. The twisted $b$-symplectic model presented here is suited for the case of laminar viscous flows in which the Reynolds number is small enough. In general, the model is good for flows of low complexity and no turbulence and for which the Stokes' Law is a valid approximation.

The fact that the twisted $b$-symplectic model can not be obtained from a cotangent lift is a geometrical result that reveals that there is essential information about the dynamics of the system contained in the fibers of the configuration space $T^*Q$. Indeed, the core of the $1$-dimensional linear model is that the fixed singularity at the $0$ fiber of the cotangent bundle $T^*Q$ makes any trajectory tend to a fixed point on the base. And in the non-linear models, the effect is analogous for any orbit that would intersect the $0$ fiber in the classical non-twisted symplectic setting.

The key features of the model are perfectly illustrated in the phase portraits of Figures \ref{fig:phasespace}, \ref{fig:phasespacequad} and \ref{fig:phasespaceperiod}. There, the orbits that intersect transversally the $0$ fiber in the classical symplectic setting are transformed into escape orbits when replacing the standard symplectic form by the twisted $b$-symplectic form. This type of orbits, which can be seen as union of trajectories, also arises in other contexts such as celestial mechanics, and has been recently investigated in the $b$-contact context by Miranda, Oms and Peralta-Salas \cite{MirandaOms21,MirandaOmsPeralta-Salas22}. See also \cite{Miranda20} for singular periodic orbits in the realm of $b$-symplectic geometry where the existence of such periodic orbits escapes the classical identification with critical points of the action functional inaugurated by Rabinowitz \cite{Rabinowitz78}. Our twisted models detect an infinite number of this type of singular orbits. This situation aligns with the conjecture of \emph{"2 or infinity"} in the number of  periodic Reeb orbits also in the singular set-up \cite{HoferWysockiZehnder98, Cristofaro-GardinerHutchingsPomerleano19}.

We have seen that dissipation naturally emerges from a singular symplectic form in the direction of the singularity. In particular, this provides a simple model for uni-dimensional friction. This can be generalized to multiple dimensions with arbitrary external field by including an additional dimension to describe the physical time and energy. The dissipation and then the singularity must be applied on this extra-dimension, while the Hamiltonian must be re-scaled accordingly to the dissipation coefficient. Therefore, any $d$-dimensional system with an external potential and a global dissipation given by a fixed dissipation factor can be naturally described by an Hamiltonian dynamics on a $(d+1)$-dimensional $b$-symplectic manifold.

\bibliographystyle{alpha}
\bibliography{bibfile}

\end{document}